\documentclass{article}
\usepackage[letterpaper,top=2cm,bottom=2cm,left=3cm,right=3cm,marginparwidth=1.75cm]{geometry}
\usepackage{authblk}
\usepackage{enumerate}
\usepackage{amsmath}
\usepackage{amssymb,latexsym,mathrsfs}
\usepackage{amsthm}
\usepackage{color}
\usepackage{cancel}
\usepackage{graphicx}
\usepackage{cite}
\usepackage{longtable}
\usepackage{amscd}
\usepackage{lineno}
\usepackage{booktabs} 
\usepackage{makecell}
\usepackage{hyperref}

\newcommand{\numberset}{\mathbb}
\newcommand{\N}{\numberset{N}}
\newcommand{\Z}{\numberset{Z}}

\newcommand{\Oo}{\mathcal{O}}

\newcommand{\K}{\numberset{K}}
\newcommand{\F}{\numberset{F}}

\newcommand{\Xx}{\mathcal{X}}
\newcommand{\Yy}{\mathcal{Y}}

\newcommand{\PP}{\mathbb{P}}
\newcommand{\Aut}{\mathrm{Aut}}
\newcommand{\Div}{\mathrm{Div}}

\newtheorem{thm}{Theorem}[section]

\newtheorem{prop}[thm]{Proposition}
\newtheorem{cor}[thm]{Corollary}

\newtheorem{ex}[thm]{Example}
\newtheorem{rem}[thm]{Remark}

\newtheorem{defn}[thm]{Definition}

\title{On QC and GQC algebraic geometry codes}
\author[1]{Matteo Bonini}
\author[2]{Arianna Dionigi} 
\author[3]{Francesco Ghiandoni}
\affil[1]{Aalborg University, Department of Mathematical Sciences. mabo@math.aau.dk}
\affil[2]{University of Florence, Department of Mathematics and Computer Science. arianna.dionigi@unifi.it}
\affil[3]{University of Primorska, Department of Mathematics. francesco.ghiandoni@famnit.upr.si}
\date{}

\begin{document}
	
	\maketitle
	
	\begin{abstract}
		
		We present new constructions of quasi-cyclic (QC) and generalized quasi-cyclic (GQC) codes from algebraic curves. Unlike previous approaches based on elliptic curves, our method applies to curves that are Kummer extensions of the rational function field, including hyperelliptic, norm–trace, and Hermitian curves. This allows QC codes with flexible co-index. Explicit parameter formulas are derived using known automorphism-group classifications.
	\end{abstract}
	
	\medskip
	
	
	
	\medskip

	\section*{Introduction}
	Error-correcting codes are fundamental tools in digital communication, designed to protect information transmitted over noisy channels by introducing redundancy that enables error detection and correction. Since the seminal works of Golay, Hamming, and Shannon in the late 1940s, coding theory has evolved rapidly in response to the growing demands of communication systems and advances in mathematical methods. Among the most studied families are cyclic and quasi-cyclic (QC) codes, valued for their strong algebraic structure, compact generator matrices, and efficient encoding and decoding algorithms \cite{ling2001algebraic}.Despite their long history and extensive study, it remains unknown whether cyclic codes are asymptotically good, a problem that has been open for more than fifty years \cite{martinez2006class}. In contrast, quasicyclic codes were shown over four decades ago to achieve asymptotically good parameters \cite{chen1969some}.
	
	In recent years, QC codes have attracted renewed attention due to their relevance in post-quantum cryptography, where structured linear codes play an important role in the design of secure and efficient cryptographic primitives \cite{melchor2018hamming}. This increased interest motivates the search for new constructions that offer both strong parameters and structural flexibility.
	
	QC codes generalize cyclic codes by requiring invariance under a shift of multiple positions, known as the co-index. This symmetry leads to highly structured generator matrices, so that only a small number of rows must be stored, significantly reducing memory requirements and accelerating encoding and decoding procedures. More generally, generalized quasi-cyclic (GQC) codes may be viewed as tuples of generating polynomials extending the cyclic paradigm.
	
	Traditionally, constructions of QC codes from algebraic curves have focused primarily on elliptic curves \cite{DingTong}, exploiting their automorphism groups to define the code structure. However, this approach is inherently limited: automorphisms of elliptic curves have order at most \(6\), restricting both the achievable co-index and the resulting design flexibility. By contrast, curves with larger cyclic automorphism groups naturally support QC codes whose co-index can be chosen as any divisor of the group order, providing a valuable degree of freedom in practical implementations.
	
	In this paper, we present a general construction of algebraic--geometry QC and GQC codes over \(\mathbb{F}_{q^r}\) (with \(r \ge 1\)), based on the action of automorphism groups on the sets of rational points of suitable algebraic curves. We give explicit constructions arising from several classes of curves, with particular emphasis on those obtained from Kummer extensions of rational function fields. Our framework also yields concrete examples illustrating how to construct QC codes with a wide variety of co-indices.
	
	Our approach applies in particular to a broad family of curves that includes hyperelliptic curves, the norm--trace curve, the Hermitian curve, and certain well-known quotients of the latter. Unlike the elliptic case, these curves often possess automorphisms of significantly higher order, frequently forming cyclic groups. Consequently, QC codes exist with co-index equal to any divisor of the group order. Moreover, cyclic codes themselves arise as a special case when a single orbit of the group action is used, highlighting a subtle but meaningful connection between cyclic symmetry and geometric constructions.
	
	Several of the curves considered are maximal over \(\mathbb{F}_{q^2}\), leading to QC codes with asymptotically good parameters in both rate and relative minimum distance. More generally, increasing the genus tends to improve the Singleton defect, further enhancing performance. In the hyperelliptic setting, examples of maximal curves with cyclic automorphism groups likewise provide strong parameters together with flexibility in the choice of genus and co-index.
	
	\section{Preliminaries}

	Unless stated otherwise, we will consider $q = p^{h}$ for a prime $p$ and an integer $h > 0$, and denote by $\F_q$ the finite field with $q$ elements. 
	
	In several situations we will also consider algebraic curves defined over extensions $\F_{q^{r}}$ of the base field, for example when dealing with maximal curves, Deligne–Lusztig curves, or other families naturally defined over a field larger than $\F_{q}$. When this happens, the corresponding extensions of scalars for the curve and its function field are understood implicitly.

	\subsection{Algebraic curves} 
	Let $\K$ be the algebraic closure of $\F_q$. We denote by $\PP^{n}$ (resp. $\mathbb{A}^n$) the $n$-dimensional projective (resp. affine) space over $\K$, and by $\textnormal{PG}(n,q)$ (resp. $\textnormal{AG}(n,q)$ or $\F_q^n$) the set of $\F_q$-rational points of $\PP^{n}$ (resp. $\mathbb{A}^n$).
	
	Let $F(X,Y)\in \F_q[X,Y]$ be a polynomial defining an affine, geometrically irreducible plane algebraic curve $\Xx$ over $\F_q$ via the equation $F(X,Y)=0$.
	The degree of $\mathcal{X}$ is defined as $\deg(\mathcal{X})=\deg(F)$. Its function field over $\K$ is denoted by $\K(\Xx)$, and the subfield of $\F_q$-rational functions is denoted by $\F_q(\Xx)$. We write $\Div_q(\Xx)$ for the group of divisors of $\Xx$ defined over $\F_q$, and $\Aut_\mathbb{K}(\Xx)$ (resp. $\Aut(\Xx)$) for the group of $\mathbb{K}$-automorphisms (resp. $\F_q$-automorphisms) of $\Xx$. 
	Let $P=(u,v)\in \mathbb{A}^2(\K)$ be a point in the affine plane, and consider the expansion
	\[
	F(X+u,Y+v)=F_0(X,Y)+F_1(X,Y)+F_2(X,Y)+\cdots+F_n(X,Y),
	\]
	where each $F_i$ is either zero or a homogeneous polynomial of degree $i$, and $n$ is the degree of $F$. The \emph{multiplicity} of $P$ on $\mathcal{X}$, denoted by $m_P(\mathcal{X})$ or $m_P(F)$, is defined as the smallest integer $m$ such that $F_m\neq 0$. The plane curve given by the equation $F_m=0$ is called the \emph{tangent cone} of $\mathcal{X}$ at $P$, and any linear component of the tangent cone is referred to as a \emph{tangent} to $\mathcal{X}$ at $P$.
	
	The point $P$ lies on the curve $\mathcal{X}$ if and only if $m_P(\mathcal{X})\ge 1$. If this is the case, $P$ is said to be a \emph{simple} (or nonsingular) point of $\mathcal{X}$ when $m_P(\mathcal{X})=1$, whereas it is a \emph{singular} point if $m_P(\mathcal{X})>1$. An analogous definition applies to points at infinity, namely to points of $\mathcal{X}$ lying on the line at infinity in its projective closure. We denote by $\mathrm{Sing}(\mathcal{X})$ the set of all singular points of $\mathcal{X}$, both affine and at infinity.
	
	Let $H \le \Aut_\mathbb{K}(\Xx)$ and consider the fixed field
	\[
	\K(\Xx)^{H} = \{ f \in \K(\Xx) \mid \sigma(f) = f \text{ for all } \sigma \in H \}.
	\]
	Let $\Xx/H$ denote the non singular projective curve with function field $\K(\Xx)^{H}$. The extension $\K(\Xx) : \K(\Xx)^{H}$ is Galois with group $H$, and we write
	\[
	\pi \colon \Xx \longrightarrow \Xx/H
	\]
	for the associated quotient morphism.
	
	For a point $P \in \Xx$, the stabilizer $H_{P}$ is the subgroup of $H$ fixing $P$, and the ramification index at $P$ is $e_{P} = |H_{P}|$. Points with $e_{P} > 1$ are ramification points, and their images under $\pi$ are branch points.
	
	The group $H$ acts on the points of $\Xx$ and partitions them into orbits
	\[
	\Oo(P) = \{ \sigma(P) \mid \sigma \in H \}.
	\]
	By the orbit–stabilizer theorem,
	\[
	|\Oo(P)| \cdot |H_{P}| = |H|.
	\]
	An orbit is long if $|\Oo(P)| = |H|$ and short otherwise. A point $Q \in \Xx/H$ is a branch point if and only if the orbit lying above $Q$ is short. It is known that $H$ has only finitely many short orbits, and that the covering $\pi$ is unramified if and only if every orbit is long.

	If the genus $g(\Xx) \ge 2$, then $\Aut_\mathbb{K}(\Xx)$ is finite (see for example \cite[Chapter 11]{HKT}). The classical Hurwitz bound,
	\[
	|\Aut_{\mathbb{K}}(\Xx)| \le 84 (g(\Xx) - 1),
	\]
	fails in general in positive characteristic when $p$ divides $|\Aut_{\mathbb{K}}(\Xx)|$. Nevertheless, stronger bounds are known under additional hypotheses. In particular, the following result holds for abelian subgroups.
	
	\begin{thm}[{\cite[Theorem 11.79]{HKT}}]
		\label{thm uppbond abel autom group}
		Let $H$ be an abelian subgroup of $\Aut_\mathbb{K}(\Xx)$. Then
		\[
		|H| \le
		\begin{cases}
			4 g(\Xx) + 4, & \text{if } p \ne 2,\\[2mm]
			4 g(\Xx) + 2, & \text{if } p = 2.
		\end{cases}
		\]
	\end{thm}

	\subsection{Linear Codes}
	
	\label{sottosez AG codes}   
	
	We recall here some basic definitions in coding theory (for a complete exposition of these concepts the reader is referred to \cite{HP}).
	
	Let $k,n$ be two positive integers such that $k \leq n$. A $q$-ary linear code of dimension $k$ and length $n$ (or an $[n,k]$ code over $\F_q$) is a $k$-dimensional vector subspace $\mathcal{C}\subseteq \F_q^n$.  
	The elements of $\mathcal{C}$ are called codewords. A generator matrix of $\mathcal{C}$ is any $k \times n$ matrix $G$ whose rows form a basis of $\mathcal{C}$ over $\F_q$.
	
	Classically, a linear code is endowed with the Hamming metric.
	
	For $x\in\F_q^n$, the Hamming weight of $x$ is the number of nonzero coordinates of $x$, i.e.,
	\[
	\mathrm{wt}(x) = |\mathrm{supp}(x)|,
	\]
	where $\mathrm{supp}(x)=\{i : x_i\neq 0\}$.  
	The \emph{Hamming distance} between $v,w\in\F_q^n$ is defined as
	\[
	d(v,w)=\mathrm{wt}(v-w).
	\]
	
	The minimum (Hamming) distance of a linear code $\mathcal{C}$ is defined as
	\[
	d(\mathcal{C}) = \min\{\mathrm{wt}(c) \mid c\in \mathcal{C},\, c\neq 0\}.
	\]
	
	One of the most fundamental bounds in coding theory is the Singleton bound.
	
	\begin{prop}[Singleton bound]
		Let $\mathcal{C}$ be an $[n,k]$ linear code with minimum distance $d$. Then
		\[
		d \leq n - k + 1.
		\]
	\end{prop}
	
	The Singleton defect of an $[n,k,d]$ code $\mathcal{C}$ is defined as
	\[
	s(C) = n - k + 1 - d.
	\]
	A code for which $s(\mathcal{C}) = 0$ is called a Maximum Distance Separable (MDS) code; equivalently, $\mathcal{C}$ meets the Singleton bound with equality.

	MDS codes play a central role in both classical and modern coding theory, appearing for example in Reed--Solomon codes and in several constructions for distributed storage.
	A code $\mathcal{C}$ with Singleton defect $s(\mathcal{C})=1$ is called a almost-MDS (AMDS) code. A code $\mathcal{C}$ is said to be near-MDS if both $\mathcal{C}$ and $\mathcal{C}^\perp$ are AMDS.
	
	Both MDS and AMDS codes satisfy very stringent combinatorial constraints: for example, their parameters impose tight relations between the minimum distances of $\mathcal{C}$ and $\mathcal{C}^\perp$, and only very restricted families of lengths $n$ and alphabet sizes $q$ can support such codes. 
	Related families of interest arise from algebraic-geometry (AG) codes constructed on algebraic curves, where the Singleton defect is upper bounded by the genus of the curve; this provides large classes of codes whose parameters are similarly governed by geometric or combinatorial restrictions, and which interpolate between MDS/AMDS behaviour and more general regimes.

	\subsection{Algebraic Geometry Codes}
	
	Let $\mathcal{X}$ be a projective, smooth, geometrically irreducible algebraic curve defined over the finite field $\F_q$, and let $\F_q(\mathcal{X})$ denote its function field. For a divisor 
	\[
	G = \sum_{P \in \mathcal{X}} n_P P \in \mathrm{Div}_q(\mathcal{X}),
	\]
	the associated \emph{Riemann–Roch space} is defined as
	\[
	\mathscr{L}(G)
	= \{ f \in \F_q(\mathcal{X})^\times \mid (f) + G \ge 0 \} \cup \{0\},
	\]
	where $(f)$ is the principal divisor of $f$. This is a finite-dimensional $\F_q$-vector space, and we denote its dimension by $\ell(G)$.
	
	Let $P_1, \dots, P_n$ be $n$ distinct $\F_q$-rational points of $\mathcal{X}$, and let $G \in \mathrm{Div}_q(\mathcal{X})$ be a divisor such that $v_{P_i}(G) = 0$ for every $i=1,\dots,n$, so that $\mathrm{Supp}(G)$ is disjoint from $\{P_1,\dots,P_n\}$. Consider the evaluation map
	\[
	\begin{split}
		\mathrm{ev} = \mathrm{ev}_{P_1,\ldots,P_n} :
		\mathscr{L}(G) &\longrightarrow \F_q^n,\\
		f &\longmapsto \left( f(P_1), \ldots, f(P_n) \right).
	\end{split}
	\]
	If we define
	\[
	D := P_1 + \cdots + P_n,
	\]
	the \emph{algebraic geometry (AG) code associated with $D$ and $G$} is
	\[
	\mathcal{C}(D,G) := \mathrm{ev}(\mathscr{L}(G)) \subseteq \F_q^n .
	\]
	
	\medskip
	
	Standard properties of AG codes (see \cite[Section II.2]{st}) give the following parameters:
	\begin{equation*}
		n = \deg(D), \qquad 
		k = \ell(G) - \ell(G - D), \qquad
		d \ge n - \deg(G).
	\end{equation*}
	Moreover, by the Riemann–Roch theorem, if
	\[
	n > \deg(G) > 2g - 2,
	\]
	then $G$ is non-special and hence
	\[
	k = \deg(G) - g + 1.
	\]

	\section{Construction of (generalized) quasi-cyclic AG codes}

	Quasi-cyclic codes arise as linear codes that are invariant under a structured blockwise cyclic shift. We formalize this below.

	\begin{defn}
		Let $q$ be a prime power and let $\ell,m \in \mathbb{Z}_{>0}$.  
		A linear code $\mathcal{C} \subseteq \F_q^{\,\ell m}$ is said to be
		\emph{quasi-cyclic of index $\ell$ (and co-index $m$)} if it is invariant
		under a blockwise cyclic shift of length $m$.
		
		More precisely, write each codeword $c \in \F_q^{\,\ell m}$ as
		\[
		c = (c_0,c_1,\dots,c_{\ell-1}),
		\qquad
		c_j = (c_{0,j},c_{1,j},\dots,c_{m-1,j}) \in \F_q^{\,m},
		\]
		so that $\F_q^{\,\ell m}$ is viewed as the direct product of $\ell$ blocks,
		each of length $m$.
		
		Let $\tau : \F_q^{\,m} \to \F_q^{\,m}$ denote the cyclic shift
		\[
		\tau(c_{0,j},c_{1,j},\dots,c_{m-1,j})
		=
		(c_{m-1,j},c_{0,j},\dots,c_{m-2,j}),
		\]
		and define the operator $T : \F_q^{\,\ell m} \to \F_q^{\,\ell m}$ by
		\[
		T(c_0,c_1,\dots,c_{\ell-1})
		=
		(\tau(c_0),\tau(c_1),\dots,\tau(c_{\ell-1})).
		\]
		
		Then $\mathcal{C}$ is a QC code of index $\ell$ if
		\[
		T(\mathcal{C}) = \mathcal{C}
		\quad
		\text{(equivalently, $T^k(\mathcal{C})=\mathcal{C}$ for all $k\in\mathbb{Z}$).}
		\]
	\end{defn}
	
	\noindent The following definition generalizes this notion by allowing the block lengths to vary.
	
	\begin{defn}
		Let $q$ be a prime power and let $m_0,m_1,\dots,m_{\ell-1} \in \mathbb{Z}_{>0}$.
		A linear code
		\[
		\mathcal{C} \subseteq \F_q^{\,m_0} \times \F_q^{\,m_1} \times \cdots \times \F_q^{\,m_{\ell-1}}
		\]
		is called a \emph{generalized quasi-cyclic (GQC) code with block lengths
			$(m_0,m_1,\dots,m_{\ell-1})$} if it is invariant under a blockwise cyclic shift.
		
		Explicitly, write each codeword as
		\[
		c = (c_0,c_1,\dots,c_{\ell-1}),
		\qquad
		c_j = (c_{0,j},c_{1,j},\dots,c_{m_j-1,j}) \in \F_q^{\,m_j}.
		\]
		For each $j$, let $\tau_j : \F_q^{\,m_j} \to \F_q^{\,m_j}$ denote the cyclic shift
		\[
		\tau_j(c_{0,j},c_{1,j},\dots,c_{m_j-1,j})
		=
		(c_{m_j-1,j},c_{0,j},\dots,c_{m_j-2,j}),
		\]
		and define
		\[
		T(c_0,c_1,\dots,c_{\ell-1})
		=
		(\tau_0(c_0),\tau_1(c_1),\dots,\tau_{\ell-1}(c_{\ell-1})).
		\]
		
		Then $\mathcal{C}$ is a GQC code if
		\[
		T(\mathcal{C}) = \mathcal{C}
		\quad
		\text{(equivalently, $T^k(\mathcal{C})=\mathcal{C}$ for all $k\in\mathbb{Z}$).}
		\]
	\end{defn}

	\begin{rem}
		Note that if $\mathcal{C}$ is a GQC code of length $(m_1,\dots,m_s)$ with $m_1=\cdots=m_s$ then $\mathcal{C}$ is a quasi cyclic code with lenght $s\cdot m_1$.
	\end{rem}
	
	From now on, we will focus on algebraic curves without affine singular points and having a unique point at infinity, denoted as $P_\infty$. 
	
	\begin{thm}\label{costruzione} Let $\mathcal{X}$ be a curve of genus $g$ over a finite field $\mathbb{F}_q$. Let $\sigma \in \textnormal{Aut}(\Xx)$ be an automorphism of $\mathcal{X}$ fixing  $P_\infty$. Let $\mathcal{O}_1,\dots,\mathcal{O}_m$ be $m$ distinct long orbits of $\mathcal{X}(\F_q)$ under the action of $\sigma$, and let the divisors $D_1$ and $D_2$ be defined as follows
		\[D_1=\sum_{i=1}^{m}\sum_{P_i\in \mathcal{O}_i}P_i,
		\]
		\[
		D_2=t P_\infty\]
		where $2g-2<t<m \cdot ord(\sigma)$. The algebraic geometry code
		\begin{eqnarray*}
			\mathcal{C}(D_1,D_2)&=&(f(P_1),f(\sigma(P_1)),\dots,f(\sigma^{ord(\sigma)-1}(P_1)), \dots,f(P_m),f(\sigma(P_m)),\dots,f(\sigma^{ord(\sigma)-1}(P_m)))
		\end{eqnarray*}
		where $f\in\mathscr{L}(D_2)$ is an $\left[m \cdot ord(\sigma),t+1-g,d\geq n-t\right]_q$ QC code with co-index $ord(\sigma)$.
	\end{thm}

	\begin{proof}
		From \cite[Corollary 2.2.3]{st} it follows that if  
		\[
		2g - 2 < t < m \cdot \operatorname{ord}(\sigma),
		\]
		then \(k = t + 1 - g\) and \(d \ge n - t\). Since \(\sigma\) fixes \(P_{\infty}\), we have \(\sigma^{-1}(P_{\infty}) = P_{\infty}\). For every \(f \in \mathscr{L}(D_2)\),
		\[
		v_{P_{\infty}}(\sigma(f)) 
		= v_{\sigma^{-1}(P_{\infty})}(f) 
		\ge t,
		\]
		and for every \(P \in \mathcal{X} \setminus \{P_{\infty}\}\),
		\[
		v_P(\sigma(f)) = v_{\sigma^{-1}(P)}(f) \ge 0.
		\]
		Thus \(\sigma(f) \in \mathscr{L}(D_2)\).
		
		Let \(\operatorname{ord}(\sigma)=r\), and let \(\{f_1, \dots, f_k\}\) be a basis of \(\mathscr{L}(D_2)\). The generator matrix \(G\) of the code \(\mathcal{C}(D_1, D_2)\) is
		\[
		G=
		\begin{pmatrix}
			f_1(P_1) & \cdots & f_1(\sigma^{r-1}(P_1)) & f_1(P_2) & \cdots & f_1(\sigma^{r-1}(P_2)) & \cdots & f_1(P_m) & \cdots & f_1(\sigma^{r-1}(P_m)) \\
			f_2(P_1) & \cdots & f_2(\sigma^{r-1}(P_1)) & f_2(P_2) & \cdots & f_2(\sigma^{r-1}(P_2)) & \cdots & f_2(P_m) & \cdots & f_2(\sigma^{r-1}(P_m)) \\
			\vdots & \cdots & \vdots & \vdots & \cdots &\vdots & \cdots & \vdots & \cdots & \vdots \\
			f_k(P_1) & \cdots & f_k(\sigma^{r-1}(P_1)) & f_k(P_2) & \cdots & f_k(\sigma^{r-1}(P_2)) & \cdots & f_k(P_m) & \cdots & f_k(\sigma^{r-1}(P_m))
		\end{pmatrix}
		\]
		\[
		= (ev(f_1), \dots, ev(f_k))^{T}.
		\]
		To prove that \(\mathcal{C}(D_1,D_2)\) is quasi-cyclic, we must show that if  
		\[
		\gamma = (\gamma_1, \dots, \gamma_n) = \sum_{i=1}^k c_i \, ev(f_i)
		\]
		belongs to \(\mathcal{C}(D_1, D_2)\), then the vector
		\[
		\gamma' = (\gamma_r, \gamma_1, \dots, \gamma_{r-1},
		\gamma_{2r}, \gamma_{r+1}, \dots, \gamma_{2r-1},
		\dots, \gamma_n, \gamma_{n-r+1}, \dots, \gamma_{n-1})
		\]
		also belongs to \(\mathcal{C}(D_1, D_2)\).
		
		We know that \(\gamma' \in \mathcal{C}(D_1,D_2)\) if and only if  
		\[
		\gamma' = (g(P_1), g(\sigma(P_1)), \dots, g(\sigma^{r-1}(P_1)),
		g(P_2), \dots, g(\sigma^{r-1}(P_2)), \dots,
		g(P_m), \dots, g(\sigma^{r-1}(P_m)))
		\]
		for some \(g \in \mathscr{L}(D_2)\). Equivalently, for each \(i \in \{1,\dots,r\}\),
		\[
		\gamma_i = g(\sigma^{i}(P_1)),
		\]
		and similarly for each block of length \(r\).
		
		In fact, for each \(i \in \{1,\dots,r\}\) and every \(\ell \in \{0,\dots,m-1\}\),
		\[
		\gamma_{i + \ell r}
		= \sum_{j=1}^k c_j f_j(\sigma^{\,i-1}(P_{\ell+1}))
		= \sum_{j=1}^k c_j f_j(\sigma^{-1}(\sigma^{i}(P_{\ell+1})))
		= \sum_{j=1}^k c_j \sigma(f_j)(\sigma^{i}(P_{\ell+1})).
		\]
		
		The claim follows by choosing  
		\[
		g = \sum_{j=1}^k c_j \sigma(f_j),
		\]
		which lies in \(\mathscr{L}(D_2)\) by the first part of the proof.
	\end{proof}

	\begin{thm}\label{costruzionegen} Let $\mathcal{X}$ be a curve of genus $g$ defined over a finite field $\mathbb{F}_q$. Let $\sigma$ be an automorphism of $\mathcal{X}$ that fixes  $P_\infty$. Let $\mathcal{O}_1,\dots,\mathcal{O}_s$ be $s$ distinct non-trivial orbits of $\mathcal{X}$ under the action of $\sigma$, each of length $m_1,\dots,m_s,$ $m_i \mid \textnormal{ord}(\sigma).$ Consider the divisors $D_1$ and $D_2$ defined by
		$$D_1=\sum_{i=1}^{s}\sum_{P_i\in \mathcal{O}_i}P_i,\,\,D_2=t\overline{P}$$
		with $2g-2<t<%m \cdot ord(\sigma)
		m_1+\dots + m_s$. Then the algebraic geometry code
		\[
		\mathcal{C}(D_1,D_2)=(f(P_1),f(\sigma(P_1)),\dots,f(\sigma^{ord(\sigma)-1}(P_1)),\dots, f(P_m),f(\sigma(P_m)),\dots,f(\sigma^{ord(\sigma)-1}(P_m))),
		\]
		where $ f\in\mathscr{L}(D_2)$, is an $\left[n,k,d\right]_q$ GQC code with co-index $(m_1,\dots,m_s)$, $n=m_1+ \dots + m_s $, $k=t+1-g$ and $d\geq n-t$.
	\end{thm}
	
	\begin{proof} The proof is a straight generalization of that of Theorem \ref{costruzione}.
	\end{proof}

	\section{Curves defined by separated polynomials}
	In this section we show how Theorem~\ref{costruzione} can be used to construct a broad and explicit family of QC algebraic geometry codes. We also discuss, in certain cases, the structure of the corresponding automorphism groups.
	
	Assume that $\mathrm{char}(\mathbb{F}_q)\neq 2$.
	Let $\mathcal{X}$ be a projective plane curve defined over $\mathbb{F}_{q^r}$, with $q$ a power of a prime $p$, given affinely by
	\begin{equation}\label{eq gen kummer}
		X^{m} = B(Y),
	\end{equation}
	where $B(Y)\in \mathbb{F}_{q^r}[Y]$ is a separable polynomial of degree $d$, $m\ge 2$, $d\ge 1$, $m\not\equiv 0 \pmod p$, and $\gcd(m,d)=1$. We additionally assume that $\deg(\mathcal{X})\ge 3$.
	
	Curves of this type are commonly referred to as \emph{cyclic}, \emph{superelliptic}, or \emph{Kummer} extensions of the rational function field. They admit an automorphism $\tau\in \mathrm{Aut}(\mathcal{X})$ of order $m$, defined by
	\[
	(x,y)\longmapsto (\xi x,\, y), \qquad \xi^{m}=1,
	\]
	such that the cyclic group $\langle \tau\rangle$ is normal in $\mathrm{Aut}(\mathcal{X})$, and the quotient curve $\mathcal{X}/\langle \tau\rangle$ has genus zero.
	
	Some relevant properties of $\mathcal{X}$ are summarized in the theorem below; see \cite[Proposition~3.7.3]{st} for an equivalent formulation in terms of function fields. Further special cases include the situation where $B(Y)$ is a linearized polynomial, treated in \cite[Section~12.1]{HKT} and \cite[Lemma~2.1]{boninimontanuccizini}, and the case $d=\deg(B(Y))>m$, discussed in \cite[Section~5.1]{shaska2019}.

	\begin{thm} \label{thm riassunto prop famiglia generale}
		The curve $\mathcal{X}:X^m-B(Y)=0$ is an absolutely irreducible curve with no affine singular points.
		\begin{itemize}
			\item [(i)] If $| m-d |=1, $ then $\mathcal{X}$ is non-singular.
			\item [(ii)]  \begin{itemize} 
				\item [(a)] If $m>d+1,$ then $P_{\infty}=(0:0:1)$ is an $(m-d)$-fold point of $\mathcal{X}.$
				\item [(b)] If $m<d+1,$ then $P_{\infty}=(0:1:0)$ is an $(d-m)$-fold point of $\mathcal{X}.$ 
				\item [(c)] In both cases, $P_{\infty}$ is the centre of only branch of $\mathcal{X};$ also, $P_{\infty}$ is the unique infinite point of $\mathcal{X}.$
			\end{itemize}
			\item [(iii)] The natural projection $\pi: \mathcal{X} \rightarrow  \mathcal{X}/\langle\tau\rangle \cong \mathbb{P}^1,$ $\pi(x,y)=y,$ has degree $m.$ Such cover is branched exactly at the roots $\alpha_1,\dots,\alpha_d$ of $B(Y)$. In particular $\mathcal{X}$ has genus $g=\frac{(m-1)(d-1)}{2}.$
			\item [(iv)] Let $\mathbb{K}=\overline{\F_q},$ and let $\mathbb{K}(x,y)=\mathbb{K}(\Xx),$ where $x^m=B(y),$ denote the function field of $\mathcal{X}.$ Then $\mathbb{K}(y) \subset \mathbb{K}(x,y)$ is a cyclic Galois extension of degree $m$ and $\textnormal{Gal}(\mathbb{K}(x,y) /\mathbb{K}(y))=\langle\tau\rangle.$
			\item [(v)] Let $G:=\textnormal{Aut}_{\mathbb{K}}(\mathbb{K}(x,y)).$ Then the quotient group $\overline{G}:=G/\langle\tau\rangle $ is isomorphic  to a subgroup of $\textnormal{PGL}(2,\mathbb{K})$ preserving the set $\{P_1,\dots,P_{d},P_\infty\}$ of all Weierstrass points of $\mathcal{X},$  where $P_i=(0,\alpha_i)$ for $i=1,\dots,d.$
		\end{itemize}
	\end{thm}
	The group $\overline{G}$ and the subfield $\K(\Xx)^G$ of $\K(\Xx)$ are called, respectively, the \emph{reduced automorphism group} of $\Xx$ and the \emph{invariant} of $G$ (or $G$-\emph{invariant}) in $\K(\Xx)$. By Galois theory and Theorem~\ref{thm riassunto prop famiglia generale}(iii), the extension $\K(\Xx)^G \subseteq \K(y)$ is rational. In particular, any generator of $\K(\Xx)^G$ is a $G$-invariant. Hence $\K(\Xx)^G=\K(z)$, where the invariant $z$ can be chosen as a rational function in $y$ of degree $|\overline{G}|$.
	
	For a general plane curve defined by an affine equation $F(X,Y)=0$, determining an explicit $G$-invariant as a rational function in $x$ and $y$ is typically a difficult problem, and may be computationally infeasible when the defining polynomial $F(X,Y)$ is too complicated or when the action of the generators of $G$ is not given by linear transformations. 
	
	In contrast, for plane curves defined by equations of the form~\eqref{eq gen kummer}, an explicit expression for $z=\phi(y)$ can be obtained more readily. We refer to \cite[Lemma~12]{shaska2019} for the case $d=\deg(B(Y))>m$, and to \cite{PGUinvariant} for the case where $\Xx$ is the Hermitian curve.
	
	\begin{rem} \label{reduced automorf group}
		Since $\mathcal{X}/\langle\tau\rangle$  has genus zero and $\overline{G} = \textnormal{Aut}(\mathcal{X}/\langle\tau\rangle)$ acts on the projective line as a subgroup of  $\textnormal{PGL}(2,\mathbb{K}),$ it follows that $\overline{G}$ is isomorphic to one of the
		following: $C_n,$ $D_n,$ $A_4,$ $S_4,$ $ A_5,$ semidirect product of elementary Abelian group
		with cyclic group, PSL$(2, q^r)$ and PGL$(2, q^r),$ where \begin{itemize} 
			\setlength\itemsep{0pt}
			\setlength\parskip{0pt}
			\setlength\parsep{0pt}
			\item $C_n = \langle \sigma \rangle$, with $\sigma(y) = \zeta y$ and $\zeta$ is a primitive $n$-th root of unity. 
			\item $D_{2n} = \langle \sigma, \lambda \rangle$, with $\sigma(y) = \xi y$, $\lambda(y) = \frac{1}{y}$, and $\xi$ is a primitive $n$-th root of unity. 
			\item $A_4= \langle \sigma, \mu \rangle$, with $\sigma(y) = -y$ and $\mu(y) = i \frac{y+1}{y-1}$, with $i^2 = -1$. 
			\item $S_4 = \langle \sigma, \mu \rangle$, with $\sigma(y) = i y$ and $\mu(y) = i \frac{y+1}{y-1}$,  $i^2 = -1$. 
			
			\item $A_5= \langle \sigma, \rho \rangle$, where $\sigma(y) = \xi y$, $\rho(y) = -\frac{y+b}{b y - 1}$, $\xi$ is a primitive fifth root of unity, and $b = -i(\xi + \xi^4)$ with $i^2 = -1$.
		\end{itemize}
		see \cite{valentini}. Also, in \cite{classautomgroupsuperelliptic}  all possible automorphism groups of genus $g \ge 2$ cyclic curves defined over a finite field of characteristic $p \neq 2$ are determined (for details see also \cite[Theorem 17]{shaska2019}).
	\end{rem}
	\begin{rem}\label{rem QC da aut centrale}
		Let $m \mid q^r-1$ and $\ell:=\#\{v \in \F_{q^r}:B(v)=0\}\le d.$ Then it is easy to see that $\tau$ is defined over $\F_{q^r}$ and induces on $\mathcal{X}(\F_{q^r})$ $\ell+1$ short orbits $\{P_1\},\dots,\{P_\ell\},\{P_\infty\},$ i.e., $\tau$ fixes $\ell +1$ $\F_{q^r}$-rational points of $\Xx$. Thus Theorem \ref{costruzione} yields the existence of an $[n,t-\frac{(m-1)(d-1)}{2}+1,\ge n-t]$ QC code of co-index $m,$ where  $ n=\frac{\mathcal{X}(\F_{q^r})-\ell-1}{m},$ and $(m-1)(d-1)-2 < t < n.$  
	\end{rem}

	In what follows, we undertake a detailed analysis of several well-known families of plane curves defined by separated polynomials of the form~\eqref{eq gen kummer}. For each family, we determine the automorphism group and, in particular, investigate the action of its elements on the rational points of the curve. Building on this analysis, we derive the existence of various QC and GQC codes with significant parameters.
	We conclude the section with an example illustrating how to obtain QC  codes with a high degree of flexibility in the co-index parameter starting from a curve whose automorphism group is cyclic.
	
	\begin{ex} \label{ex genere 6 kummer}
		Let $r=1$ and assume that $q \equiv 1 \pmod{21}$. Consider the curve $\Xx \colon x^3 = y^8 - y$ defined over $\F_q$, which has genus $g=7$. It is straightforward to verify that $\Xx$ admits a cyclic automorphism group $G$ of order $21$ such that $G = G_{P_\infty} = \langle \sigma \rangle$, where the automorphism $\sigma$ is given by $(x,y) \mapsto (\epsilon x,\xi y)$, with $\epsilon,\xi \in \F_q^\ast$ nontrivial elements satisfying $\epsilon^3 = \xi^7 = 1$. By Theorem~\ref{thm uppbond abel autom group}, the group $G$ is the largest cyclic automorphism group of $\Xx$.
		
		A direct computation shows that $\sigma$ induces on $\Xx(\F_q)\setminus\{P_\infty\}$ two short orbits of lengths $1$ and $7$, respectively, together with $\frac{\#\Xx(\F_q)-9}{21}$ long orbits.
		
		For $q=127$, a computation with \textsc{Magma}~\cite{magma} shows that the automorphism $\sigma(x,y)=(19x,2y)$ generates seven orbits of length $21$, denoted by $\mathcal{O}_1,\dots,\mathcal{O}_7$. Consequently, by Theorem~\ref{costruzione}, the AG code $\mathcal{C}(D_1,D_2)$, where $D_1=\sum_{i=1}^7\sum_{P_i \in \mathcal{O}_i} P_i$ and $D_2=tP_\infty$, is a $[147,t-6,\ge 147-t]_{127}$ QC code of co-index $21$ for $12<t<147$.
		
		Moreover, by considering the automorphism $\sigma^3(x,y)=(x,\xi y)$, one can also construct QC codes of co-index $7$. In this case, the condition on the field size can be relaxed to $q \equiv 1 \pmod{7}$.
	\end{ex}

	\section{Hyperelliptic curves}
	Kummer extensions of function fields, together with the associated family of superelliptic curves, constitute the main objects underlying the curves considered in our construction. The geometric curves arising from such extensions fit naturally into a broader framework, with hyperelliptic curves representing perhaps the most classical example, as they provide a setting in which several structural features can be described explicitly.
	
	When the exponent in~\eqref{eq gen kummer} equals \(m=2\), one obtains the class of \emph{hyperelliptic curves}, which may be viewed as the simplest nontrivial superelliptic curves. Throughout this section we assume that \(\mathrm{char}(\mathbb{F}_q)\neq 2\).
	
	Hyperelliptic curves play a central role in algebraic geometry as double covers of the projective line defined by a distinguished involution. This symmetry often leads to rich and frequently computable automorphism groups, allowing for a detailed understanding of their geometry and arithmetic. Beyond their theoretical relevance, the efficient arithmetic of their Jacobians has made hyperelliptic curves useful in applications, particularly in cryptography, where compact representations and fast group operations are desirable.
	
	A hyperelliptic curve \(\mathcal{H}\) defined over a field \(K\) of characteristic different from \(2\) admits an affine model of the form
	\[
	\mathcal{H}:\quad x^2 = f(y),
	\]
	where \(f(y)\in K[y]\) is a separable polynomial of degree \(d\ge 5\). The separability of \(f(y)\) ensures that the curve has no affine singularities.
	
	As one can see from  $(iii)$ in Theorem \ref{thm riassunto prop famiglia generale}, the projection \((x,y)\mapsto y\) defines a morphism of degree \(2\) from \(\mathcal{H}\) onto \(\mathbb{P}^1\); hence hyperelliptic curves arise as double covers of the projective line branched precisely at the roots of \(f(y)\) (and possibly at the point at infinity, depending on the parity of \(d\)). Such curves possess a distinguished involution, called the \emph{hyperelliptic involution}, given by
	\[
	(x,y)\longmapsto (-x,y),
	\]
	which generates a central subgroup of \(\mathrm{Aut}(\mathcal{H})\).
	
	Because of this explicit structure, hyperelliptic curves provide a natural setting in which many geometric and arithmetic properties can be described concretely, and they often serve as a testing ground for constructions over finite fields, including algebraic geometry codes and the study of automorphism groups.
	
	\begin{rem}
		Consider $m=2$ and let $\mathcal{X}:x^2-f(y)=0$ be an hyperelliptic curve of genus $g,$  with a square-free polynomial $f\in \F_{q^r}[Y]$ of degree $2g+1.$ %or $2g+2.$ 
		Then it is well known that $\mathcal{X}$ has no affine singular points and the unique point at infinity $P_\infty$ has multiplicity $2g-1$. Moreover, by Theorem \ref{thm riassunto prop famiglia generale} there is only one branch centered at it. The automorphism $\tau$ of $\Xx$ is now given by $(x,y)\mapsto (-x,y),$ known as the hyperelliptic involution. This yields by Remark \ref{rem QC da aut centrale} the existence of $[n,k,d]_{q^r}$ QC codes of co-index $2,$ where $n \le \#\Yy(\F_{q^r})-2(g+1),$ $k=n-t+1$ and $d \ge n-t,$ with $2g-2 < t < n.$ \end{rem}

	Starting from hyperelliptic curves, one can construct quasi-cyclic codes with relevant parameters, both asymptotically (as the size of the base field grows) and over fixed fields. When the base field is fixed and long codes are desired, it is advantageous to consider hyperelliptic curves of high genus, as they can admit many rational points (recall the Hasse--Weil upper bound $\#\Xx(\F_q) \le q+1+2g\sqrt{q}$). On the other hand, if strong error-correcting capabilities are required, curves of low genus are often preferable, since they typically yield codes with larger minimum distance, as discussed in Subsection~\ref{sottosez AG codes}.
	
	In particular, curves of genus $2$ allow one to construct, via Theorem~\ref{costruzione}, QC codes with parameters $[n,k,d]_{q^r}$ and minimum distance $d=n-k-1$, and in some cases $d=n-k$ (see Example~\ref{es g 2 coind 10}), together with a fairly flexible choice of the block length (co-index). These aspects will be examined in greater detail in the next section.
	
	We conclude this section with several examples of QC codes constructed from a hyperelliptic curve of genus $4$. In light of Theorem~\ref{thm uppbond abel autom group}, such a curve admits the largest possible cyclic automorphism group for a fixed genus, which, from the perspective of QC codes, provides greater flexibility in the choice of the corresponding co-index.

	\begin{ex}
		Let $r=1$ and $q\equiv 1 \mod{18}$ and consider the curve $\Xx\colon x^2 = y^{9}+a$ over $\F_q$, where $a$ is a $9$th power in $\F_q$. It is straightforward to verify that $\Xx$ admits a cyclic automorphism group over $\F_q$ of order $18$, generated by the automorphism $\sigma(x,y)=(-x,\xi y)$, where $\xi \in \F_q$ is a primitive $9$th root of unity. Moreover, the action of $\sigma$ on $\Xx(\F_q)\setminus\{P_\infty\}$ consists of two short orbits of lengths $2$ and $9$, respectively, and of $\frac{\#\Xx(\F_q)-12}{18}$ long orbits. \\
		In particular, for $q=73$ the curve $\Xx$ has $66$ $\F_q$-rational points, and the automorphism $\sigma\colon (x,y)\mapsto (-x,2y)$ induces on $\Xx(\F_q)$ three orbits of length $18$. By choosing divisors $D_1$ and $D_2$ in a way completely analogous to that of Example~\ref{ex genere 6 kummer}, one can construct via Theorem \ref{costruzione} $[66,t-3,\ge 66-t]_q$ QC codes of co-index $18$, for every $6<t<66$. 
		
		Finally, since for $i \in \{2,3,6\}$ one has $\textnormal{ord}(\sigma^i)=\frac{18}{i}$ and the action of $\sigma^i$ on $\Xx(\F_q)$ is equivalent to that of $\sigma$, this construction also yields QC codes of co-indices $3$, $6$, and $9$.

	\end{ex}

	\subsection{Genus 2}
	For curves $\mathcal{X}$ of genus $2$, a comprehensive classification of their automorphism groups was developed through a sequence of works produced by Cardona, Duursma, Gonzalez, Gutierrez, Kiyavash, Lario, Quer, Rio,
	Shaska, and V\"{o}\!\!
	lklein; see \cite{duursma2005vector,cardona1999curves,cardona2007curves,shaska2004elliptic}. This research effort, involving several contributors over those years, now provides a complete description of the possible automorphism groups in genus~2.
	
	Outside characteristic $2$, the classification of automorphism groups of genus-$2$ curves shows that the possible groups are
	\[
	\mathbb{Z}_2,\ \mathbb{Z}_{10},\ \mathbb{Z}_2 \times \mathbb{Z}_2,\ D_8,\ D_{12},\ 
	\mathbb{Z}_3 \rtimes D_8,\ \mathrm{GL}(2,3),
	\]
	and an extension of $S_5$ by the hyperelliptic involution.
	
	\smallskip
	
	Since in our setting we are particularly interested in automorphism groups containing elements of various possible orders, so as to allow greater flexibility in the co-index of the corresponding quasi-cyclic codes, we focus on the cases
	\[
	\operatorname{Aut}(\mathcal{X}) \in 
	\{\mathbb{Z}_3 \rtimes D_8,\ \mathrm{GL}(2,3),\ \mathbb{Z}_{10}\}.
	\]
	For $p \neq 2$, the genus-$2$ curves whose automorphism group $G$ belongs to this list are completely determined. They are birationally equivalent to irreducible curves of the form
	\[
	x^2 = f(y),
	\]
	where
	\[
	f(y) =
	\begin{cases}
		y^6 - 1, & \text{if } G \cong \mathbb{Z}_3 \rtimes D_8,\\[4pt]
		y^5 - y, & \text{if } G \cong \mathrm{GL}(2,3),\\[4pt]
		y^5 + 1, & \text{if } G \cong \mathbb{Z}_{10}.
	\end{cases}
	\]
	
	In the first case, however, the stabilizer $G_{P_\infty}$ contains no elements of order other than $1,2,3,4,$ and $6$. Consequently, Theorem~\ref{costruzione} yields only quasi-cyclic codes of co-index $2,3,4,$ or $6$. Codes of this type were already obtained in \cite{DingTong}, starting from AG codes on elliptic curves; indeed, any automorphism of an elliptic curve fixing the point at infinity has order at most $6$, see \cite[Theorem~1]{DingTong} and \cite[Theorem~11.94]{HKT}.
	
	\smallskip
	
	On the other hand, the curves 
	\[
	\mathcal{X}\colon x^2 = y^5 + 1 
	\qquad\text{and}\qquad
	\mathcal{Y}\colon x^2 = y^5 - y
	\]
	give rise to quasi-cyclic codes of co-indices $10$ and $8$, respectively.
	
	\begin{rem} \label{rem genus 2 aut ciclico 8}
		Let $p \neq 2,5,$  $q^r \equiv 1$ (mod $8$) and $\mathcal{Y}\colon x^2 = y^5 - y.$ Then an automorphism $ \sigma \in G_{P_\infty}$  is given by  $\sigma(x,y)=(ax,a^2 y),$ where $a \in \F_{q^r}$ is a primitive $8$-th root of unity. In particular, $\operatorname{ord}(\sigma)=8$, and its action on $\mathcal{Y}(\mathbb{F}_{q^r})\setminus\{P_\infty\}$ induces two short orbits of sizes $1$ and $4$, namely $\{O\}=\{(0,0)\}$ and $\{(0,1),(0,a^2),(0,-1),(0,-a^2)\}$,  together with $\frac{\#\mathcal{Y}(\mathbb{F}_q)-6}{8}$ long orbits.
	\end{rem}
	The following proposition is now a direct consequence of the previous remark and of Theorem~\ref{costruzione}.
	
	\begin{prop}  \label{prop g 2 orbita 8}
		Let $\mathcal{Y}\colon x^2 = y^5 - y$ and let  $\sigma$ be as in Remark \ref{rem genus 2 aut ciclico 8}. Let $\mathcal{O}_1,\dots,\mathcal{O}_{m},$ $m=\frac{\#\mathcal{Y}(\mathbb{F}_{q^r})-6}{8},$ be all the possible long orbits defined by $\sigma.$ Consider the divisors $D_1=\sum_{i=1}^{m}\sum_{P_i\in \mathcal{O}_i}P_i,$  and $D_2=tP_\infty.$ Then for $2<t<\#\mathcal{Y}(\mathbb{F}_{q^r})-6,$ the algebraic geometry code $\mathcal{C} (D_1,D_2)$ is a $[\#\mathcal{Y}(\mathbb{F}_{q^r})-6,t-1,\ge \#\mathcal{Y}(\mathbb{F}_{q^r})-6 -t]_{q^r}$ QC code of co-index $8.$ 
	\end{prop}
	For the curve $\mathcal{X}\colon x^2 = y^5 + 1$, the construction is essentially analogous, with a few additional details. We present it below.
	
	\begin{rem} \label{rem genus 2 aut ciclico 10}
		Let $p \neq 2,5$ and $q^r \equiv 1$ (mod $5$). Then $G=G_{P_\infty}=\Z_2 \times Z_5=\langle \sigma \rangle,$ where $\sigma(x,y)=(-x,\xi y),$ and $\xi \in \F_{q^r}$ is a primitive $5$-th root of unity. In particular, $\operatorname{ord}(\sigma)=10$, and its action on $\mathcal{X}(\mathbb{F}_q)\setminus\{P_\infty\}$ induces two short orbits of sizes $2$ and $5$, namely $\{R_1,R_2\}=\{(1,0),(-1,0)\}$ and $\{Q_1,Q_2,Q_3,Q_4,Q_5\}$, where $Q_i=(0,-\xi^i)$, together with $\frac{\#\mathcal{X}(\mathbb{F}_{q^r})-8}{10}$ long orbits.
	\end{rem}
	\begin{prop}
		Let $\sigma$ be as in Remark \ref{rem genus 2 aut ciclico 10}. Let $\mathcal{O}_1,\dots,\mathcal{O}_{m},$ $m=\frac{\#\mathcal{X}(\mathbb{F}_{q^r})-8}{10},$ be all the possible long orbits defined by $\sigma.$ Consider the divisors $D_1=\sum_{i=1}^{m}\sum_{P_i\in \mathcal{O}_i}P_i,$ $D_2=D_1+\sum_{i=1}^5Q_i+R_1+R_2$ and $D_3=tP_\infty.$ Then for $2<t<\#\mathcal{X}(\mathbb{F}_{q^r})-8,$ the algebraic geometry code $\mathcal{C} (D_1,D_3)$ is a $[\#\mathcal{X}(\mathbb{F}_{q^r})-8,t-1,\ge \#\mathcal{X}(\mathbb{F}_{q^r})-8 -t]_{q^r}$ QC code of co-index $10.$ Also, for $2<t<\#\mathcal{X}(\mathbb{F}_q)-1,$ the algebraic geometry code $\mathcal{C} (D_2,D_3)$ is a $[\#\mathcal{X}(\mathbb{F}_{q^r})-1,t-1,\ge \#\mathcal{X}(\mathbb{F}_{q^r})-1-t]_{q^r}$ GQC code of co-index $(2,5,10,\dots,10).$
	\end{prop}
	
	Moreover, we observe that by simply considering $\sigma^2(x,y) = (x,\xi^2 y)$ and its action on the rational points of $\mathcal{X}$, one can again apply Theorem~\ref{costruzione} to obtain quasi-cyclic codes of co-index $5$.\\
	Here we give some examples, which can be computed and verified directly using \textsc{Magma}~\cite{magma}.
	
	\begin{ex}
		Let $r=1$ and $q=41.$ Then for the curve $\mathcal{Y}:x^2=y^5-y$ with $54$ $\F_q$-rational points there is an automorphism of order $8$ fixing $P_\infty$ given by $\sigma(x,y)=(3x,9y).$ By a direct computation one gets that $\sigma$ generates $6$ long orbits $\mathcal{O}_1,\dots,\mathcal{O}_6$  given by \[
		\begin{aligned}
			\mathcal{O}_1 &= \{(7,19), (21,7), (22,22), (25,34), (34,19), (20,7), (19,22), (16,34)\},\\
			\mathcal{O}_2 &= \{(34,30), (20,24), (19,11), (16,17), (7,30), (21,24), (22,11), (25,17)\},\\
			\mathcal{O}_3 &= \{(15,35), (4,28), (12,6), (36,13), (26,35), (37,28), (29,6), (5,13)\},\\
			\mathcal{O}_4 &= \{(39,5), (35,4), (23,36), (28,37), (2,5), (6,4), (18,36), (13,37)\},\\
			\mathcal{O}_5 &= \{(27,10), (40,8), (38,31), (32,33), (14,10), (1,8), (3,31), (9,33)\},\\
			\mathcal{O}_6 &= \{(40,26), (38,29), (32,15), (14,12), (1,26), (3,29), (9,15), (27,12)\}.
		\end{aligned}
		\]
		Now, the AG code $\mathcal{C}(D_1,D_2)$ defined in Proposition~\ref{prop g 2 orbita 8} is a $[48,\, t-1,\, \ge 48-t]_{41}$ QC code of co-index $8$, for $3 \le t \le 47$.

	\end{ex}
	\begin{ex} \label{es g 2 coind 10}
		Let $r=1$ and $q=31.$ Then $\#\mathcal{X}(\mathbb{F}_q)=28$, and the automorphism $\sigma(x,y)=(-x,2y)$ induces on $\mathcal{X}(\mathbb{F}_q)\setminus\{P_\infty\}$ two short orbits, namely $\{(1,0),(-1,0)\}$ and $\{(0,15),(0,23),(0,27),(0,29),(0,30)\}$, and two long orbits given by
		\[
		\mathcal{O}_1=\{(21,11),(10,26),(21,22),(10,21),(21,13),(10,11),(21,26),(10,22),(21,21),(10,13)\}
		\]
		and
		\[
		\mathcal{O}_2=\{(23,8),(8,2),(23,16),(8,4),(23,1),(8,8),(23,2),(8,16),(23,4),(8,1)\}.
		\]  Then, for $3 \le t \le 19$, the AG code $\mathcal{C}(D_1,D_3)$ is a $[20,\, t-1,\, \ge 20-t]_{31}$ QC code of co-index $10$, where the minimum distance satisfies $d = 20 - t + 1$ (that is, $\mathcal{C}(D_1,D_3)$ is an NMDS code) for $t \in \{3,17,19\}$, and $d = 20 - t$ in the remaining cases.
		Moreover, by simply using the two short orbits defined above and the automorphism $\sigma^2$, one obtains $[27,\, t-1,\, \ge 27-t]_{31}$ GQC codes of co-index $(2,5,10,10)$ and $[20,\, t-1,\, \ge 20-t]_{31}$ QC codes of co-index $5$.

	\end{ex}
	In the next subsection, we will focus on the case $r=2$ by considering a family of $\mathbb{F}_{q^2}$-maximal hyperelliptic curves. In particular, we will study the curve $\mathcal{X}_g\colon y^2 = x^{2g+1} + 1$ of genus $g$, which coincides with the curve $\mathcal{X}$ described above in the case $g=2$.

	\subsection{The $\F_{q^2}$-maximal curve $x^2=y^{2g+1}+1$ 
	}

	In \cite[Theorem 6
	]{tafaz} the author provides a complete characterization for the curve $\mathcal{X}_g:x^2-y^{2g+1}-1=0$ 
	to be $\F_{q^2}$-maximal (see also \cite[Theorem 1]{Tafaz_JPAA}). Specifically $\Xx_g$ is $\F_{q^2}$-maximal if and only if $2g+1$ divides $q+1.$
	
	\begin{rem}
		Consider $G:=\{(x,y)\mapsto (\epsilon x,\xi y) : \epsilon^2=\xi^{2g+1}=1 \}.$ It is easy to see that $G$ is an automorphism group of $\Xx_g$ fixing $P_\infty$ and  $G\cong \Z_2\times \Z_{2g+1}.$ Also, by Theorem~\ref{thm uppbond abel autom group}, the group $G$ is not properly contained in any abelian automorphism group of $\Xx$.
		This means, together with Theorems~\ref{costruzione} and \ref{costruzionegen}, that starting from AG codes arising from the curve $\mathcal{X}_g$, one can construct QC codes of co-index $\ell$, where $\ell$ is any divisor of $|G| = 4g + 2$.
		
		Infact, consider the field $\F_{q^2},$ where $q$ is an odd prime power. Let $g\in \N$ such that $2g+1 \mid (q+1).$ Then $\#\Xx_g(\F_{q^2})=q^2+2gq+1$ and $G=\langle \sigma \rangle,$ where $\sigma(x,y)=(-x,\xi y),$ with $\textnormal{ord}(\sigma)=2(2g+1).$ Also, $\sigma$ fixes $P_\infty$ and induces on $\Xx_1(\F_{q^2})\setminus \{P_\infty\}$   two short orbits of size  $2$ and $2g+1,$ respectively, and $\frac{q-1}{2}+\frac{(q-2)(q+1)}{4g+2}$ long orbits. Therefore, Theorems~\ref{costruzione} and~\ref{costruzionegen} ensure the existence of a QC code and a GQC code with parameters 
		\[
		[q^2 +2g(q-1)-3,\, t_1 - g + 1,\, \ge q^2 +2g(q-1)-3 - t_1]_{q^2}
		\quad\text{and}\quad
		[q^2 + 2gq,\, t_2 - g + 1,\, \ge q^2 + 2gq - t_2]_{q^2},
		\]
		for $2g - 2 < t_1 < q^2 - 2$ and $2g - 2 < t_2 < q^2 + 2g$, with co-indices $4g + 2$ and $(2,2g + 1, 4g + 2, \dots, 4g + 2)$, respectively.\\
		Finally, since any power of $\sigma$ has order dividing $4g+2$ and the action of such a power on $\mathcal{X}_g(\mathbb{F}_{q^2})$ is analogous to that of $\sigma$, one can construct QC and GQC codes with the parameters just described and co-indices $2\ell$ and $(2,\ell,2\ell,\dots,2\ell)$, respectively, where $\ell$ is any divisor of $2g+1$.
	\end{rem}

	\section{The norm-trace curve}
	Let \begin{equation}\label{eq normtrace}\mathcal{N}_{q,r}: x^{\frac{q^r-1}{q-1}}=y^{q^{r-1}}+y^{q^{r-2}}+\cdots+y \end{equation}be the norm-trace curve. 
	The norm-trace curve is quite famous in the literature \cite{geil2003codes}, and has recently received renewed attention due to its applications in the construction of good algebraic geometry codes and in the study of their parameters \cite{bonini2020intersections,bonini2023rational,carvalho2024decreasing}.
	It is well known that $\mathcal{N}_{q,r}$ has $q^{2r-1}+1$ $\F_{q^r}$-rational points, with a unique point at infinity $P_{\infty}=(0:0:1)$. Also, in \cite[Theorem 3.1]{boninimontanuccizini} the authors prove that for $r\geq 3$, ${\rm{Aut}}(\mathcal{N}_{q,r})={\rm{Aut}}(\mathcal{N}_{q,r})_{P_\infty}$ has order $q^{r-1}(q^r-1)$ and is a semidirect product $G\rtimes C$ where
	$$G=\left\{(x,y)\mapsto(x,y+a)\,\,|\,\,\textnormal{Tr}_{q^r|q}(a)=0\right\}\,\,\,\,{\rm{and}}\,\,\,\,C=\left\{(x,y)\mapsto(bx,b^{\frac{q^r-1}{q-1}}y)\,\,|\,\,b\in\mathbb{F}_{q^r}^*\right\}.$$
	The following proposition is a direct consequence of \cite[Theorem 3.1]{boninimontanuccizini}. Denote by $\Omega=\{(0,\overline{y}): \textnormal{Tr}_{q^r|q}(\overline{y})=0\}$ the set of the $q^{r-1}$ $\F_{q^r}$-rational points of $\mathcal{N}_{q,r}$ which are the zeros of $x.$ 
	\begin{prop}  \label{prop orbite norma traccia}
		Let $\sigma \in {\rm{Aut}}(\mathcal{N}_{q,r})=G \rtimes C.$ Then the following cases occur. \begin{itemize}
			\item If $\sigma \in G,$ then $\sigma$ acts semiregularly on $\mathcal{N}_{q,r}(\F_{q^r})\setminus \{P_\infty\}$ with $q^{2r}/p$ orbits of length $p.$
			\item If $\sigma \in C$, set $\sigma(x,y)=(bx,b^{\frac{q^r-1}{q-1}}y),$ for some $b \in \F_{q^r}^*.$ Then  $\sigma$ fixes the point $(0,0)$, $\operatorname{ord}(\sigma)=\operatorname{ord}(b)\mid q^r-1$, and $\sigma$ induces on $\Omega \setminus \{(0,0)\}$ exactly $\frac{q^{r-1}-1}{\operatorname{ord}(\gamma)}$ short orbits, each of length $\operatorname{ord}(\gamma)$, where $\gamma = b^{\frac{q^r-1}{q-1}}$. Moreover, $\sigma$ acts semiregularly on $\mathcal{N}_{q,r}(\mathbb{F}_{q^r}) \setminus \bigl(\{P_\infty\} \cup \Omega\bigr)$.
			\item If $\sigma \notin G \cup C,$ set $\sigma(x,y)=(bx,b^{\frac{q^r-1}{q-1}}y+a),$ with $a \neq 0$ and $b \neq 1.$ \begin{itemize}
				\item [(a)] If $\gamma = b^{\frac{q^r-1}{q-1}} \neq 1$, then $\operatorname{ord}(\sigma)=\operatorname{ord}(b)$ and $\sigma$ fixes the point $P=(0,\tfrac{a}{\gamma-1}) \in \Omega$, and acts on $\Omega \setminus \{P\}$ with $\frac{q^{r-1}-1}{\operatorname{ord}(\gamma)}$ short orbits of length $\operatorname{ord}(\gamma)$.
				Moreover, $\sigma$ acts semiregularly on $\mathcal{N}_{q,r}(\mathbb{F}_{q^r}) \setminus \bigl(\{P_\infty\} \cup \Omega\bigr)$.
				\item [(b)] If $\gamma=1,$ then $\operatorname{ord}(\sigma)=p\cdot \operatorname{ord}(b),$ and $\sigma$ acts on $\Omega$ with  $\frac{q^{r-1}}{p}$ short orbits of length $p$. Moreover, $\sigma$ acts semiregularly on $\mathcal{N}_{q,r}(\mathbb{F}_{q^r}) \setminus \bigl(\{P_\infty\} \cup \Omega\bigr)$.
			\end{itemize}
		\end{itemize}
	\end{prop}
	
	The next two corollaries are direct consequences of Theorems~\ref{costruzione} and~\ref{costruzionegen}, Proposition~\ref{prop orbite norma traccia}, together with the equalities
	\[
	\{\operatorname{ord}(b) : b \in \F_{q^r}^{*}\}
	=
	\{\ell \in \N : \ell \mid q^r-1\},
	\qquad
	\operatorname{ord}(\gamma)
	=
	\operatorname{ord}\!\left(b^{\frac{q^r-1}{q-1}}\right)
	=
	\gcd\!\left(\operatorname{ord}(b),\, q-1\right).
	\]
	
	\begin{cor}
		Let $r \ge 3.$ Then there exist $[n,t -\frac{q(q^{r-1}-1)^2}{2(q-1)}+1,\ge n-t]_{q^r}$ QC linear codes of co-index $m$, where $ \frac{q(q^{r-1}-1)^2}{q-1}-2 < t < n$ and $n,m$  are as follows.
		\begin{itemize}
			\item 
			$n=q^{2r-1}$ and $m=p; $ or
			\item $n=q^{2r-1}-q^{r-1}$ and $m \in \{p^\nu \ell : \nu \in \{0,1\} \ \ \textnormal{and} \ \ \ell \mid q^r-1\}.$
		\end{itemize}
	\end{cor}
	\begin{cor}
		Let $r \ge 3.$ Then there exist $[n,t -\frac{q(q^{r-1}-1)^2}{2(q-1)}+1,\ge n-t]_{q^r}$ GQC linear codes of co-index $(m_1,\dots,m_s)$, where $ \frac{q(q^{r-1}-1)^2}{q-1}-2 < t < n$ and $n,m_1,\dots,m_s$  are as follows.
		\begin{itemize}
			\item $n=q^{2r-1}-1,$  \begin{itemize}
				\item $\ell \mid q^r-1,$
				\item $s=s_1+s_2,$ where $s_1:=\frac{q^{r-1}-1}{\gcd(\ell,q-1)}$ and $s_2:=\frac{q^{2r-1}-q^{r-1}}{\ell},$
				\item $m_1=\dots=m_{s_1}=\gcd(\ell,q-1)$ and $m_{s_1+1}=\dots=m_s=\ell.$      
			\end{itemize}  
			\item $n=q^{2r-1}$
			\begin{itemize}
				\item $\ell \mid \frac{q^r-1}{q-1},$
				\item $s=s_1+s_2,$ where $s_1:=\frac{q^{r-1}}{p}$ and $s_2:=\frac{q^{2r-1}-q^{r-1}}{p\cdot \ell},$
				\item $m_1=\dots=m_{s_1}=p$ and $m_{s_1+1}=\dots=m_s=p\cdot \ell.$
			\end{itemize}
		\end{itemize}
	\end{cor}
	
	In the next section, we will focus on the case $r=2$, in which the norm--trace curve defined in Equation~\eqref{eq normtrace} coincides with the well-known Hermitian curve.
	
	\section{The Hermitian curve}
	Let $\mathcal{H}_q : Y^{q} + Y = X^{q+1}$ be the Hermitian curve over $\mathbb{F}_{q^2}$.  
	Recall that $\mathcal{H}_q$ is $\mathbb{F}_{q^2}$‑maximal; that is, it attains the Hasse--Weil upper bound and has
	\[
	q^{2} + 1 + 2g(\mathcal{H}_q)\, q = q^{3} + 1
	\]
	$\mathbb{F}_{q^2}$-rational points, where 
	\[
	g(\mathcal{H}_q)=\frac{q(q-1)}{2}.
	\]
	Moreover, $\mathcal{H}_q$ is the norm--trace curve corresponding to the extension 
	$\mathbb{F}_{q^2}/\mathbb{F}_q$; in particular, it corresponds to the norm--trace curve when $r=2$.
	It is known (see \cite[Section 4]{GSX}), that the stabilizer $G$ of $P_{\infty}$ in $\textnormal{Aut}(\mathcal{H}_q)$ is 
	$$G=S_p\rtimes C,\,\,\,\,\,S_p=\left\{\psi_{1,b,c}\right\}\,\,\, \textnormal{and}\,\,\, C=\left\{\psi_{a,0,0}\right\}$$
	where 
	\begin{equation}\label{eq psi herm}\psi_{a,b,c}:(x,y)\longmapsto(ax+b,ab^q x+a^{q+1}y+c)\end{equation} 
	with 
	$$b,c\in\mathbb{F}_{q^2},a\in\mathbb{F}_{q^2}^*\,\, \textnormal{and}\,\,c^q+c=b^{q+1}.$$
	The following result, shown in \cite{GSX}, describes the complete action of each element of $G$ on the $\mathbb{F}_{q^2}$-rational points of $\mathcal{H}_q$. For the reader’s convenience, we also include the proof below.

	\begin{thm}\label{OrbiteHermitiana} Let $\psi_{a,b,c}$ as in Equation \eqref{eq psi herm}.  If $\psi_{a,b,c}\in S_p$, then $\psi_{a,b,c}$ acts semiregularly on the affine points of $\mathcal{H}_q$. If $\psi_{a,b,c}\notin S_P$ one of the following holds:
		\begin{itemize}
			\item[(1)] If $p\,|\,{\rm{ord}}(\psi_{a,b,c})$  there are $\frac{q}{p}$ short orbits of length $p$, and $\psi_{a,b,c}$ acts semiregularly on the others $q^3-q$ affine points of $\mathcal{H}_q$.
			\item[(2)] If ${\rm{ord}}(\psi_{a,b,c})\,|\,q+1$ $\psi_{a,b,c}$ fixes $q$ affine points of $\mathcal{H}_q$ and acts semiregularly on the others.
			\item[(3)] Otherwise $\psi_{a,b,c}$ fixes one affine point of $\mathcal{H}_q$, there are $\frac{q-1}{{\rm{ord}}(a^{q+1})}$ short orbits of length ${\rm{ord}}(a^{q+1})$, and $\psi_{a,b,c}$ acts semiregularly on the others $q^3-q$ points of $\mathcal{H}_q$.
		\end{itemize}    
	\end{thm}
	
	\begin{proof} If $\psi_{a,b,c}$  belongs to $S_p$ then ${\rm{ord}}(\psi_{a,b,c})=p$ and, since the $p$-rank of $\mathcal{H}_q$ is zero (\cite[Corollary 2.5]{GarTaf}), $\psi_{a,b,c}$ acts semiregularly on the affine points of $\mathcal{H}_q$ (see \cite[Theorem 11.133]{HKT}). Thus there are $\frac{q^3}{p}$ orbits of length $p$. Otherwise, $a\neq 1$, then, by the proof of \cite[Lemma $4.1$]{GSX}, $\tau^{-1}\circ\psi_{a,b,c}\circ\tau=\psi_{a,0,c^*}$, where ${c^*}^q+c^*=0$ and $\tau=\psi_{1,e,f}$ with $e=\frac{b}{a-1}$ and $f^q+f=e^{q+1}$. Therefore, when $a\neq 1$, we can assume, up to conjugation, $b=0$ and in particular ${\rm{ord}}a\mid{\rm{ord}}\psi_{a,b,c}$. For $a\neq 1$ we distinguesh the proof in three cases as in \cite[ Lemma $4.2$]{GSX}.
		\begin{itemize}
			\item [(1)] $p\,|\,{\rm{ord}}(\psi_{a,b,c})$.\\
			By \cite[Lemma $4.2$]{GSX}, in this case there are no affine points fixed by $\psi_{a,b,c}$. By \cite[Lemma $4.1$]{GSX}, since $p\,|\,{\rm{ord}}\psi_{a,b,c}$ the only possibility is $a^{q+1}=1$ and ${\rm{ord}}\psi_{a,b,c}=p\cdot{\rm{ord}}(a)$ and so, up to conjugation $\psi_{a,b,c}(x,y)=(ax,y+c^*)$ with  ${c^*}^q+c^*=0$. Consider $i \in \{0,\dots,\rm{ord}(\psi_{a,b,c})\}$ and $P=(\overline{x},\overline{y})\in \mathcal{H}_q$ such that $\psi_{a,b,c}^i(\overline{x},\overline{y})=(a^i\overline{x},\overline{y}+ip)=(\overline{x},\overline{y}).$  Note that there are $q$ points of $\mathcal{H}_q$ with $\overline{x}=0$ and $\overline{y}$ such that $\overline{y}^q+\overline{y}=0$; on these $q$ points the action of $\psi_{a,b,c}$ gives rase to $\frac{q}{p}$ orbits of length $p$. On the others $q^3-q$ points of $\mathcal{H}_q$ $\psi_{a,b,c}$ acts semiregularly, hence there are $\frac{q^3-q}{{\rm{ord}}(\psi_{a,b,c})}$ long orbits.
			\item [(2)] $p \nmid {\rm{ord}}(\psi_{a,b,c})$ and ${\rm{ord}}(\psi_{a,b,c})\,|\,q+1$.\\
			Since $a^{q+1}=1$, by the proof of \cite[Lemma $4.1$ (ii)]{GSX}, up to conjugation, we can assume $\psi_{a,b,c}(x,y)=(ax,y)$, it follows that ${\rm{ord}}(\psi_{a,b,c})={\rm{ord}}(a)$. By \cite[ Lemma $4.2$]{GSX} in this case there are $q$ fixed point, that are the $q$ points with $x=0$ (for all these points $y$ is such that $y^q+y=0$). On the others $q^3-q$ affine points of $\mathcal{H}_q$ $\psi_{a,b,c}$ acts semiregularly. Then there are $\frac{q^3-q}{{\rm{ord}}(\psi_{a,b,c})}$ long orbits.
			\item [(3)] ${\rm{ord}}(\psi_{a,b,c})\,|\,q+1$.\\
			By the proof of \cite[Lemma $4.1$, (ii) ]{GSX}, in this case, up to conjugation, $\psi_{a,b,c}(x,y)=(ax,a^{q+1}y)$ and hence ${\rm{ord}}(\psi_{a,b,c})={\rm{ord}}(a)$. Here, by \cite[ Lemma $4.2$]{GSX}, there is just one fixed point, that is $(0,0)$. There are $q-1$ points of $\mathcal{H}_q$ with $x=0$ and different from $(0,0)$; on these points the action of $\psi_{a,b,c}$ gives rase to $\frac{q-1}{{\rm{ord}}(a^{q+1})}$ orbits of length ${\rm{ord}}(a^{q+1})$, where ${\rm{ord}}(a^{q+1})=\frac{{\rm{ord}}(a)}{{\rm{gcd}}({\rm{ord}}(a),q+1)}$. On the others affine points of $\mathcal{H}_q$ with $x\neq 0$ $\psi_{a,b,c}$ acts semiregularly, and so there are $\frac{q^3-q}{{\rm{ord}}(a)}$ orbits of length ${\rm{ord}}(a)$.
			
		\end{itemize}
		
	\end{proof}

	\begin{thm}\label{HermitianCurve}
		Let $\psi_{a,b,c}$ be an automorphism of the Hermitian curve $\mathcal{H}_q$ belonging to the stabilizer of $P_\infty$.  
		Let $\mathcal{O}_1,\dots,\mathcal{O}_m$ be the long orbits of $\psi_{a,b,c}$.  
		Define the divisors
		\[
		D_1 = \sum_{i=1}^{m}\, \sum_{P \in \mathcal{O}_i} P, 
		\]\[
		D_2 = t P_\infty,
		\]
		where 
		\[
		q(q-1)-2 < t < m\cdot \mathrm{ord}(\psi_{a,b,c}).
		\]
		Then the algebraic‑geometry code $\mathcal{C}(D_1,D_2)$ is a quasi‑cyclic code of co‑index $\mathrm{ord}(\psi_{a,b,c})$.
		Moreover:
		\begin{itemize}
			\item If $\psi_{a,b,c} \in S_p$, then $\mathcal{C}(D_1,D_2)$ is a 
			\[
			[q^3,\, t-\frac{q(q-1)}{2}+1,\, \geq n-t]_{q^2}
			\]
			QC code.
			
			\item If $\psi_{a,b,c} \notin S_p$, then one of the following holds:
			\begin{itemize}
				\item[(1)] If $p \mid \mathrm{ord}(\psi_{a,b,c})$, then $\mathcal{C}(D_1,D_2)$ is a 
				\[
				[q^3 - q,\, t-\frac{q(q-1)}{2}+1,\, \geq n-t]_{q^2}
				\]
				QC code.
				
				\item[(2)] If $\mathrm{ord}(\psi_{a,b,c}) \mid (q+1)$, then $\mathcal{C}(D_1,D_2)$ is a 
				\[
				[q^3 - q,\, t-\frac{q(q-1)}{2}+1,\, \geq n-t]_{q^2}
				\]
				QC code.
				
				\item[(3)] Otherwise, $\mathcal{C}(D_1,D_2)$ is also a 
				\[
				[q^3 - q,\, t-\frac{q(q-1)}{2}+1,\, \geq n-t]_{q^2}
				\]
				QC code.
			\end{itemize}
		\end{itemize}
	\end{thm}

	\begin{proof}
		The proof is a straigthforth consequence of Theorem \ref{costruzione} and Theorem \ref{OrbiteHermitiana}.   \end{proof}
	
	\subsection{A quotient of the Hermitian curve with a large cyclic automorphism group}
	In this subsection we take into account a family of quotients of the Hermitian curve admitting a large cyclic automorphism group.
	
	Such a family consists of curves $\mathcal{C}_m$ of equation
	$$y^q-y=x^m$$
	over $\F_{q^2},$ where $m>1$, $\gcd(q,m)=1$, $p\neq 3$ and $m \mid (q+1)$. 
	The curves $\mathcal{C}_m$ can be obtained as the quotient of the Hermitian curve $\mathcal{H}_q$
	over \(\mathbb{F}_{q^2}\) with respect to the cyclic subgroup of $\mathrm{PGU}(3,q)$ consisting of the automorphisms
	$(x,y) \mapsto (\lambda x, y),$ with $\lambda^m = 1.$
	For further details see \cite{GSX}.
	
	These curves have a cyclic automorphism group of order $N=pm$, that is $G=\langle\eta\rangle$ where 
	$$\eta(x)=\zeta_mx,\,\,\,\eta(y)=y+1$$
	with $\zeta_m$ a primitive $m$-th root of unity; furthermore the genus of these curves is 
	$$g=\dfrac{(q-1)(m-1)}{2}.$$
	Also, as a consequence of a result of Serre, cited in \cite[Proposition 6]{lachaud} and proved in \cite[Proposition 5]{aubryperret}, $\mathcal{C}_m$ is an $\F_{q^2}$-maximal curve, i.e., $\#\mathcal{C}_m(\F_{q^2})=q^2+1+(q-1)(m-1)q=mq(q-1)+(q+1).$ 
	
	It is not hard to see that $\eta$ fixes the point $P_\infty$ and acts with $q/p$ short orbits of length $p$ on the set
	\[
	\Omega := \{(0,y) : y \in \mathbb{F}_q\},
	\]
	and acts semiregularly on the remaining $mq(q-1)$ $\mathbb{F}_{q^2}$-rational points.
	Therefore the following result is a direct consequence of Theorems \ref{costruzione} and \ref{costruzionegen}.
	\begin{thm}
		Let $q,m>1$ where $q$ is a power of a prime $p,$ with $p \neq 3,$ $\gcd(p,m)=1$ and $m \mid (q+1).$ Then there exist $[mq(q-1),t -\frac{(q-1)(m-1)}{2}+1,\ge mq(q-1)-t]_{q^2}$ QC linear codes of co-index $pm$, where $ \frac{q(q^{r-1}-1)^2}{q-1}-2 < t < n.$ \\
		Also, there exist $[mq(q-1)+q,t -\frac{(q-1)(m-1)}{2}+1,\ge mq(q-1)-t]_{q^2}$ GQC linear codes of co-index $(m_1,\dots,m_s),$ where $s=\frac{q}{p}+\frac{q(q-1)}{p}=\frac{q^2}{p},$ $m_1=\dots=m_{\frac{q}{p}}=p,$ and $m_{\frac{q}{p}+1}=\dots=m_{\frac{q^2}{p}}=pm.$
	\end{thm}

	\bibliographystyle{abbrv}
	\bibliography{BDG_AGQC}

@article{aubryperret,
	title={Divisibility of zeta functions of curves in a covering},
	author={Aubry, Yves and Perret, Marc},
	journal={Archiv der Mathematik},
	volume={82},
	number={3},
	pages={205--213},
	year={2004},
	publisher={Springer}
}

@article{martinez2006class,
	title={Is the class of cyclic codes asymptotically good?},
	author={Mart{\'\i}nez-P{\'e}rez, Conchita and Willems, Wolfgang},
	journal={IEEE transactions on information theory},
	volume={52},
	number={2},
	pages={696--700},
	year={2006},
	publisher={IEEE}
}

@article{ling2001algebraic,
	title={On the algebraic structure of quasi-cyclic codes. I. Finite fields},
	author={Ling, San and Sol{\'e}, Patrick},
	journal={IEEE Transactions on Information Theory},
	volume={47},
	number={7},
	pages={2751--2760},
	year={2001},
	publisher={IEEE}
}

@article{boninimontanuccizini,
	title={On plane curves given by separated polynomials and their automorphisms},
	author={Bonini, Matteo and Montanucci, Maria and Zini, Giovanni},
	journal={Advances in Geometry},
	volume={20},
	number={1},
	pages={61--70},
	year={2020},
	publisher={De Gruyter}
}

@article{chen1969some,
	title={Some results on quasi-cyclic codes},
	author={Chen, CL and Peterson, W Wesley and Weldon Jr, EJ},
	journal={Information and Control},
	volume={15},
	number={5},
	pages={407--423},
	year={1969},
	publisher={Elsevier}
}

@article{DingTong,
	title={Quasi-cyclic NMDS codes},
	author={Tong, Hongxi and Ding, Yang},
	journal={Finite Fields and Their Applications},
	volume={24},
	pages={45--54},
	year={2013},
	publisher={Elsevier}
}

@book{GarTaf,
	title={Certain maximal curves and Cartier operators},
	author={Garcia, Arnaldo and Tafazolian, Saeed},
	year={2007},
	publisher={Nieders{\"a}chsische Staats-und Universit{\"a}tsbibliothek}
}

@article{GSX,
	title={On subfields of the Hermitian function field},
	author={Garcia, Arnaldo and Stichtenoth, Henning and Xing, Chao-Ping},
	journal={Compositio Mathematica},
	volume={120},
	number={2},
	pages={137--170},
	year={2000},
	publisher={London Mathematical Society}
}

@book{HKT,
	title={Algebraic curves over a finite field},
	author={Hirschfeld, James William Peter and Korchm{\'a}ros, G{\'a}bor and Torres, Fernando and Torres, Fernando},
	volume={20},
	year={2008},
	publisher={Princeton University Press}
}

@article{lachaud,
	title={Sommes d’Eisenstein et nombre de points de certaines courbes alg{\'e}briques sur les corps finis},
	author={Lachaud, Gilles},
	journal={CR Acad. Sci. Paris},
	volume={305},
	number={01},
	year={1987}
}

@article{shaska2019,
	title={From hyperelliptic to superelliptic curves},
	author={Malmendier, Andreas and Shaska, Tony},
	journal={Albanian Journal of Mathematics},
	volume={13},
	number={1},
	pages={107--200},
	year={2019},
	publisher={Research Institute of Science and Technology (RISAT)}
}

@article{duursma2005vector,
	title={The vector decomposition problem for elliptic and hyperelliptic curves},
	author={Duursma, I.},
	journal={J. Ramanujan Mathematical Society},
	volume={20},
	number={1},
	pages={59--76},
	year={2005}
}

@article{cardona1999curves,
	title={On curves of genus 2 with Jacobian of {GL2}-type},
	author={Cardona, Gabriel and Gonz{\'a}lez, Josep and Lario, Joan-Carles and Rio, Anna},
	journal={manuscripta mathematica},
	volume={98},
	number={1},
	pages={37--54},
	year={1999},
	publisher={Springer}
}

@article{cardona2007curves,
	title={Curves of genus 2 with group of automorphisms isomorphic to $D_8$ or $D_{12}$},
	author={Cardona, Gabriel and Quer, Jordi},
	journal={Transactions of the American Mathematical Society},
	volume={359},
	number={6},
	pages={2831--2849},
	year={2007}
}

@article {PGUinvariant,
	AUTHOR = {Gatti, Barbara and Ghiandoni, Francesco and Korchm\'aros,
		G\'abor},
	TITLE = {The invariant of {${\rm PGU}(3, q)$} in the {H}ermitian
		function field},
	JOURNAL = {Art Discrete Appl. Math.},
	FJOURNAL = {The Art of Discrete and Applied Mathematics},
	VOLUME = {8},
	YEAR = {2025},
	NUMBER = {1},
	PAGES = {Paper No. 1.09, 8},
	ISSN = {2590-9770},
	MRCLASS = {14H37 (11G20 14H05)},
	MRNUMBER = {4879693},
	MRREVIEWER = {Saeed\ Tafazolian},
}

@inproceedings{shaska2004elliptic,
	title={Elliptic subfields and automorphisms of genus 2 function fields},
	author={Shaska, Tanush and V{\"o}lklein, Helmut},
	booktitle={Algebra, Arithmetic and Geometry with Applications: Papers from Shreeram S. Abhyankar’s 70th Birthday Conference},
	pages={703--723},
	year={2004},
	organization={Springer}
}

@article{magma,
	title={The Magma computational algebra system},
	author={Cannon, John and Steel, Allan and others},
	journal={Software available online (magma. maths. usyd. edu. au)},
	year={2005}
}

@article{classautomgroupsuperelliptic,
	title={Automorphism groups of cyclic curves defined over finite fields of any characteristics},
	author={Sanjeewa, R},
	journal={Albanian Journal of Mathematics},
	volume={3},
	number={4},
	pages={131--160},
	year={2009},
	publisher={Research Institute of Science and Technology (RISAT)}
}

@article{st,
	title={On automorphisms of geometric Goppa codes},
	author={Stichtenoth, Henning},
	journal={Journal of Algebra},
	volume={130},
	number={1},
	pages={113--121},
	year={1990},
	publisher={Elsevier}
}

@article{tafaz,
	title={A note on certain maximal hyperelliptic curves},
	author={Tafazolian, Saeed},
	journal={Finite Fields Their Appl.},
	volume={18},
	number={5},
	pages={1013--1016},
	year={2012}
}

@article{Tafaz_JPAA,
	title={A family of maximal hyperelliptic curves},
	author={Tafazolian, Saeed},
	journal={Journal of Pure and Applied Algebra},
	volume={216},
	number={7},
	pages={1528--1532},
	year={2012},
	publisher={Elsevier}
}

@article{valentini,
	title={A Hauptsatz of LE Dickson and Artin-Schreier extensions.},
	author={Madan, Manohar L and Valentini, Robert C},
	year={1980},
	journal={Journal f\"ur  die reine und angewandte Mathematik},
	publisher={Walter de Gruyter, Berlin/New York Berlin, New York}
}

@book{HP,
	title={Fundamentals of error-correcting codes},
	author={Huffman, W Cary and Pless, Vera},
	year={2010},
	publisher={Cambridge university press}
}

@article{bonini2023rational,
	title={Rational points on cubic surfaces and AG codes from the Norm--Trace curve},
	author={Bonini, Matteo and Sala, Massimiliano and Vicino, Lara},
	journal={Annali di Matematica Pura ed Applicata (1923-)},
	volume={202},
	number={1},
	pages={185--208},
	year={2023},
	publisher={Springer}
}

@article{bonini2020intersections,
	title={Intersections between the norm-trace curve and some low degree curves},
	author={Bonini, Matteo and Sala, Massimiliano},
	journal={Finite Fields and Their Applications},
	volume={67},
	pages={101715},
	year={2020},
	publisher={Elsevier}
}

@article{geil2003codes,
	title={On codes from norm--trace curves},
	author={Geil, Olav},
	journal={Finite fields and their Applications},
	volume={9},
	number={3},
	pages={351--371},
	year={2003},
	publisher={Elsevier}
}

@article{carvalho2024decreasing,
	title={Decreasing norm-trace codes},
	author={Carvalho, C{\'\i}cero and L{\'o}pez, Hiram and Matthews, Gretchen},
	journal={Designs, Codes and Cryptography},
	volume={92},
	number={5},
	pages={1143--1161},
	year={2024},
	publisher={Springer}
}

@article{melchor2018hamming,
	title={Hamming quasi-cyclic (HQC)},
	author={Melchor, Carlos Aguilar and Aragon, Nicolas and Bettaieb, Slim and Bidoux, Lo{\i}c and Blazy, Olivier and Deneuville, Jean-Christophe and Gaborit, Philippe and Persichetti, Edoardo and Z{\'e}mor, Gilles and Bourges, I},
	journal={NIST PQC Round},
	volume={2},
	number={4},
	pages={13},
	year={2018}
}

\end{document}